\pdfoutput=1
\documentclass[11pt]{article}
\usepackage[T1]{fontenc}
\usepackage{geometry} 
\usepackage[english]{babel}
\usepackage{cite}
\usepackage{amsmath}
\usepackage{amsthm}
\usepackage{amsfonts} 
\usepackage{amssymb} 
\usepackage{xspace}
\usepackage{mathtools} 
\usepackage{braket} 
\usepackage[linesnumbered,vlined]{algorithm2e}
\usepackage{xcolor} 
\usepackage{footnote} 
\usepackage{environ} 
\usepackage{suffix} 
\usepackage{apptools}
\usepackage{hyperref}
\usepackage{float}
\usepackage[capitalise]{cleveref} 
\usepackage{mathrsfs}

\makesavenoteenv{tabular}
\makesavenoteenv{table}

\makeatletter
\newtheorem*{rep@theorem}{\rep@title}
\newcommand{\newreptheorem}[2]{%
	\newenvironment{rep#1}[1]{%
		\def\rep@title{#2 \ref{##1}}%
		\begin{rep@theorem}}%
		{\end{rep@theorem}}}
\makeatother

\usepackage{thmtools} 
\usepackage{thm-restate} 

\DeclareMathOperator{\poly}{poly}

\crefname{question}{Question}{Questions}
\crefname{theorem}{Theorem}{Theorems}
\Crefname{lemma}{Lemma}{Lemmas}
\Crefname{figure}{Figure}{Figures}
\Crefname{claim}{Claim}{Claims}
\Crefname{observation}{Observation}{Observations}

\newtheorem{theorem}{Theorem}
\newtheorem{lemma}{Lemma}

\newtheorem{claim}{Claim}
\newtheorem{definition}{Definition}
\newtheorem{corollary}{Corollary}

\newtheorem{observation}{Observation}
\newtheorem{question}{Question}

\theoremstyle{remark}

\theoremstyle{definition}

\newcommand{\ID}{\mathsf{ID}}

\newcommand{\CCC}{\mathcal{C}}
\newcommand{\DDD}{\mathcal{D}}

\newcommand{\diam}{\mathsf{diameter}}

\newcommand{\congestion}{\mathsf{congestion}}
\newcommand{\dilation}{\mathsf{dilation}}

\newcommand{\dist}{d}

\newcommand{\OPT}{\mathsf{OPT}}

\newcommand{\ip}[1]{\left}

\newcommand{\congest}{\ensuremath{\mathsf{CONGEST}}\xspace}

\newcommand{\kbroadcast}{$k$\textsc{-broadcast}\xspace}

\newcommand\bigO[1]{\ensuremath{{O}(#1)}}
\WithSuffix\newcommand\bigO*[1]{\ensuremath{{O}\left(#1\right)}}
\newcommand\tildeBigO[1]{\ensuremath{{\tilde{{O}}}(#1)}}
\WithSuffix\newcommand\tildeBigO*[1]{\ensuremath{{\tilde{{O}}}\left(#1\right)}}
\newcommand\littleO[1]{\ensuremath{{o}(#1)}}
\WithSuffix\newcommand\littleO*[1]{\ensuremath{{o}\left(#1\right)}}
\newcommand\tildeLittleO[1]{\ensuremath{{\tilde{{o}}}(#1)}}
\WithSuffix\newcommand\tildeLittleO*[1]{\ensuremath{{\tilde{{o}}}\left(#1\right)}}
\newcommand\bigOmega[1]{\ensuremath{{\Omega}(#1)}}
\WithSuffix\newcommand\bigOmega*[1]{\ensuremath{{\Omega}\left(#1\right)}}
\newcommand\tildeBigOmega[1]{\ensuremath{{\tilde{{\Omega}}}(#1)}}
\WithSuffix\newcommand\tildeBigOmega*[1]{\ensuremath{{\tilde{{\Omega}}}\left(#1\right)}}
\newcommand\littleOmega[1]{\ensuremath{{\omega}(#1)}}
\WithSuffix\newcommand\littleOmega*[1]{\ensuremath{{\omega}\left(#1\right)}}
\newcommand\tildeLittleOmega[1]{\ensuremath{{\tilde{{\omega}}}(#1)}}
\WithSuffix\newcommand\tildeLittleOmega*[1]{\ensuremath{{\tilde{{\omega}}{\left(#1\right)}}}}
\newcommand\bigTheta[1]{\ensuremath{{\Theta}(#1)}}
\WithSuffix\newcommand\bigTheta*[1]{\ensuremath{{\Theta}\left(#1\right)}}
\newcommand\tildeTheta[1]{\ensuremath{{\tilde{{\Theta}}}(#1)}}
\WithSuffix\newcommand\tildeTheta*[1]{\ensuremath{{\tilde{{\Theta}}\left(#1\right)}}}

\makeatletter
\newcommand*{\whp}{%
    \@ifnextchar{.}%
        {w.h.p}%
        {w.h.p.\@\xspace}%
}
\makeatother

\newcommand{\polylog}{\poly\log}
\newcommand{\tmix}{\tau_{\operatorname{mix}}}

\RestyleAlgo{boxruled}
\LinesNumbered
\let\oldnl\nl
\newcommand{\nonl}{\renewcommand{\nl}{\let\nl\oldnl}}

\SetKwFunction{AssignHelpers}{Assign-Helpers}
\SetKwFunction{TransmitSkeleton}{Transmit-Skeleton}
\SetKwFunction{LearnNeighborhood}{Learn-Neighborhood}
\SetKwFunction{LearnHelpers}{Learn-Helpers}
\SetKwFunction{ReassignSkeletons}{Reassign-Skeletons}

\NewEnviron{killcontents}{}

\title{Fast Broadcast in Highly Connected Networks}

\author{\hspace{2cm} Shashwat Chandra\footnote{National University of Singapore. Email: shashwatchandra@u.nus.edu} \and Yi-Jun Chang\footnote{National University of Singapore. Email: cyijun@nus.edu.sg} \and Michal Dory\footnote{University of Haifa. Email: mdory@ds.haifa.ac.il} \hspace{2cm} \and Mohsen Ghaffari\footnote{Massachusetts Institute of Technology. Email: ghaffari@mit.edu} \and Dean Leitersdorf\footnote{National University of Singapore. Email: dean.leitersdorf@gmail.com}}

\date{}

\begin{document}
\begin{titlepage}
\maketitle
\thispagestyle{empty}

\begin{abstract}
We revisit the classic \emph{broadcast} problem, wherein we have $k$ messages, each composed of $O(\log{n})$ bits, distributed arbitrarily across a network. The objective is to broadcast these messages to all nodes in the network. In the distributed \congest model, a textbook algorithm solves this problem in $O(D+k)$ rounds, where $D$ is the diameter of the graph.
While the $O(D)$ term in the round complexity is unavoidable---given that $\Omega(D)$ rounds are necessary to solve broadcast in \emph{any graph}---it remains unclear whether the $O(k)$ term is needed in all graphs. In cases where the minimum cut size is one, simply transmitting messages from one side of the cut to the other would require $\Omega(k)$ rounds. However, if the size of the minimum cut is larger, it may be possible to develop faster algorithms.
This motivates the exploration of the broadcast problem in networks with high \emph{edge connectivity}.

In this work, we present a \emph{simple} randomized distributed algorithm for performing $k$-message broadcast in $O(((n+k)/\lambda)\log n)$ rounds in any $n$-node simple graph with edge connectivity $\lambda$. When $k = \Omega(n)$, our algorithm is \emph{universally optimal}, up to an $O(\log n)$ factor, as its complexity nearly matches an information-theoretic $\Omega(k/\lambda)$ lower bound that applies to all graphs, even when the network topology is known to the algorithm.

The setting $k = \Omega(n)$ is particularly interesting because several fundamental problems can be reduced to broadcasting $\Omega(n)$ messages. Our broadcast algorithm finds several applications in distributed computing, enabling $O(1)$-approximation for all distances and $(1+\epsilon)$-approximation for all cut sizes in $\tilde{O}(n/\lambda)$ rounds. 
\end{abstract}

\end{titlepage}

{ 
\tableofcontents
}
\thispagestyle{empty} 
\newpage

\setcounter{page}{1}

\section{Introduction}\label{sec:intro}

In this work, we provide a new tool for fast information dissemination in distributed networks that exploits the edge connectivity of the network. Our algorithm works in the standard distributed \congest model~\cite{peleg2000distributed} where the communication network is represented as an $n$-node graph $G=(V,E)$. In this model, the nodes communicate in synchronous rounds, where in each round, each node can exchange a message of $O(\log{n})$ bits with each of its neighbors.
We focus on the classic broadcast problem, which is a central building block in many distributed algorithms. 

\begin{definition} [\kbroadcast]
    Given $k$ messages $\mathcal{M}$ of $O(\log n)$ bits, where each message $m \in \mathcal{M}$ is originally stored in some node in the graph $G$, where each node can hold an arbitrary number of messages initially, the \kbroadcast problem requires ensuring that all nodes in the graph know all messages in $\mathcal{M}$. 
\end{definition}

A textbook distributed algorithm for \kbroadcast takes $O(D+k)$ rounds \cite{very_old_k_broadcast,peleg2000distributed}, for graphs with diameter $D$, by first collecting the messages at the root of a BFS tree and then broadcasting them to the network. Since the diameter of the tree is $O(D)$, a complexity of $O(D+k)$ rounds can be obtained via pipelining.
This bound is \emph{existentially} optimal in the sense that there are networks where $\Omega(D+k)$ rounds are required to solve broadcast. 
A closer look shows that the $O(D)$ term in the round complexity is \emph{universally} optimal, in the sense that $\Omega(D)$ rounds are required to solve broadcast in \emph{any graph}, as there are nodes at a distance $\Omega(D)$ from each other.
However, it is not clear that the $O(k)$ term is needed in all graphs. If the minimum cut size is one, then just sending the information from one side of the cut to the other would take $\Omega(k)$ rounds, but if the size of the minimum cut is larger, then it may be possible to obtain faster algorithms. 
This motivates the study of the broadcast problem in graphs with high \emph{edge connectivity}, where the edge connectivity $\lambda$ of a graph is defined as the size of the minimum cut in the graph. This work aims to investigate the following fundamental question for unweighted simple graphs.

\begin{question}\label{question1}
    Can \kbroadcast be solved in $o(D+k)$ rounds in graphs with high edge connectivity?
\end{question}

\paragraph{Broadcast via tree packing.} The $O(D+k)$ round complexity for the broadcast problem comes from routing $k$ messages over one tree with diameter $O(D)$. This leads to a high \emph{congestion}, as we need to route many messages over the same edges of the tree. To obtain a faster algorithm, a natural approach is to route the messages over many edge-disjoint spanning trees. This allows us to parallelize the computation and reduce congestion. A beautiful work of Tutte \cite{tutte1961problem} and Nash-Williams \cite{nash1961edge} from 1961 showed that any graph with edge connectivity $\lambda$ has a collection of $\lfloor \lambda/2 \rfloor$ edge-disjoint spanning trees, see also \cite{kundu1974bounds}. We call such a collection a \emph{tree packing}.

In a \emph{fractional} tree packing, each tree has a weight, and for each edge $e$, the total weight of trees that contains $e$ is at most one.  
A {fractional} tree packing with parameters similar to that of Tutte~\cite{tutte1961problem} and Nash-Williams~\cite{nash1961edge} can be computed efficiently in \congest: Censor-Hillel, Ghaffari and Kuhn showed a distributed algorithm that decomposes a graph with edge connectivity $\lambda$ into fractionally edge-disjoint weighted spanning trees with total weight $\lfloor \lambda/2 \rfloor (1- \epsilon)$ in $O(D+\sqrt{n \lambda})$ rounds~\cite{DBLP:conf/podc/Censor-HillelGK14}. 
Using the lower bound technique from~\cite{sarma2012distributed}, Censor-Hillel, Ghaffari and Kuhn showed that $\tilde{\Omega}(D+\sqrt{n/\lambda})$ rounds are needed to compute such a decomposition~\cite{DBLP:conf/podc/Censor-HillelGK14}. Such decompositions can indeed reduce the \emph{congestion} of a broadcast algorithm. 
However, they do not provide a bound on the \emph{diameter} of the spanning trees computed. In the worst case, the trees could have diameter $\Omega(n)$ even if the diameter $D$ of the original graph is small, and then even sending one message over the trees results in an $\Omega(n)$ complexity, which can be much higher than the diameter $D$ of the graph. 
 
To circumvent this issue, Ghaffari~\cite{DBLP:conf/icalp/Ghaffari15} studied \emph{universally optimal} algorithms for the broadcast problem, whose goal is to find an algorithm that achieves the best possible complexity for \emph{any graph}. In his work~\cite{DBLP:conf/icalp/Ghaffari15}, Ghaffari showed how to construct a fractional tree packing with total weight $\Omega(k/(\OPT \log{n}))$ and diameter $O(\OPT \log{n})$, where $\OPT$ is the optimal round complexity for broadcasting $k$ messages of the underlying graph $G$ when the topology of $G$ is known to the algorithm. 
Once the tree packing is computed, any subsequent \kbroadcast instance on $G$ can be solved in $\tilde{O}(\OPT)$ rounds, which is \emph{the best possible} up to a polylogarithmic factor!  However, a significant drawback of Ghaffari's algorithm~\cite{DBLP:conf/icalp/Ghaffari15} is that computing the tree packing already costs $\tilde{O}(D+k)$ rounds, so this algorithm does not really break the $O(D+k)$ round complexity of the textbook broadcast algorithm. Ghaffari's work~\cite{DBLP:conf/icalp/Ghaffari15} naturally left open the following questions.

\begin{question}\label{g1}
Is it possible to improve the round complexity of the tree packing computation to $o(D + k)$, or \emph{even} to $O(\OPT) \cdot n^{o(1)}$ or $\tilde{O}(\OPT)$?
\end{question}

Even if the above question can be answered affirmatively, it is not clear how much advantage we obtain over the $O(D+k)$ round complexity of the textbook broadcast algorithm, so the following question is important.

\begin{question}\label{g2}
Is it possible to quantitatively determine an approximate value of $\OPT$  as a simple function of the underlying graph $G$?
\end{question}

In a recent work~\cite{DBLP:conf/icalp/ChuzhoyPT20}, Chuzhoy, Parter, and Tan studied low-diameter tree packings on graphs with small diameter $D$. They showed that for any parameter $\eta \in [\lambda]$, there is a randomized distributed algorithm that computes a collection of $\lambda$ spanning trees of diameter ${O}((101 \lambda \eta^{-1} \ln{n})^D)$ in $\tilde{O}((101 \lambda \eta^{-1} \ln{n})^D)$ rounds in such a way that each edge in the graph appears in at most $O(\eta \log n)$ trees. In the \emph{centralized} setting, a collection of $\lfloor \lambda/2 \rfloor$ spanning trees of diameter $O((101 \lambda \ln{n})^D)$  where each edge in the graph appears in at most two trees can be constructed in polynomial time. 
In general, the diameter of these tree packings can still be $\Omega(n)$ even when the diameter $D$ of the underlying graph $G$ is moderate.

These tree packings have implications for information dissemination in graphs with very small diameter: In any graph with diameter $D = o(\log n)$, it was shown in~\cite{DBLP:conf/icalp/ChuzhoyPT20} that it is possible to send $K$ bits of information from a node $s$ to another node $t$ in $\tilde{O}(K^{1-1/{(D+1)}}+K/\lambda)$ rounds.

\subsection{Our Contribution}\label{sect:contribution}
In this work, we answer \cref{question1,g1,g2}. We present a surprisingly \emph{simple} distributed algorithm for \kbroadcast in $\tilde{O}((n+k)/\lambda)$ rounds in any $n$-node simple graph with edge connectivity $\lambda$, addressing \cref{question1}. 

 \begin{restatable}{theorem}{broadcast} \label{thrm:main} \kbroadcast can be solved \whp in $O((n\log n)/\delta+(k\log n)/\lambda)$ rounds in any $n$-node simple graph $G=(V, E)$ with edge connectivity $\lambda$ and minimum degree $\delta$. 
\end{restatable}

It always holds that the minimum degree $\delta \geq \lambda$, so the round complexity of our broadcast algorithm is $\tilde{O}((n+k)/\lambda)$.
\cref{thrm:main} is based on a new approach to partition the network $G$ into $\Omega(\lambda / \log n)$ edge-disjoint spanning subgraphs of diameter $O((n\log n)/\delta)$ without communication. The decomposition allows us to solve any subsequent \kbroadcast instance efficiently.

 \begin{restatable}{theorem}{treepackinga}\label{thm:tree_packing}
 Let $G$ be any simple graph with edge connectivity $\lambda$ and minimum degree $\delta$.  
 Partition $G$ into $\lambda' = \lambda / (C \log n)$ edge-disjoint  subgraphs $G_1=(V, E_1)$, \ldots, $G_{\lambda'}=(V, E_{\lambda'})$ by putting each edge of $G$ in a uniformly random $G_i$ independently, then $G_i$ is a spanning subgraph with $\diam(G_i)=O((C n\log n)/\delta)$ for all $i \in [\lambda']$ with probability $1 - n^{-\Omega(C)}$.   
\end{restatable}

The decomposition of \cref{thm:tree_packing} can be computed without communication in the sense that we may simply let each edge join one of the $\lambda' = \Omega(\lambda / \log n)$ subgraphs uniformly at random and independently. In particular, for each edge $e = \{u,v\}$ with $\ID(u) > \ID(v)$, we can let $u$ decide which subgraph to which $e$ belongs.

By spending extra $O((n\log n)/\delta)$ rounds to perform a BFS in parallel for all the edge-disjoint spanning subgraphs in \cref{thm:tree_packing}, we may obtain a \emph{tree packing} of $\Omega(\lambda / \log n)$ edge-disjoint spanning trees with diameter $O((n\log n)/\delta)$.  \cref{thm:tree_packing} is not only interesting from a distributed computing point of view but also provides a new \emph{existential} bound for low-diameter tree packings. To the best of our knowledge, such a low-diameter tree packing was not known before. 
The parameters of our tree packings nearly match the \emph{existential} lower bounds from \cite{DBLP:conf/wdag/GhaffariK13}, which showed a family of graphs with $n$ nodes and diameter $O(\log{n})$ where in any tree packing the diameter of all trees is $\Omega(n/\lambda)$, except at most $O(\log{n})$ trees that may have a smaller diameter.

The decomposition of \cref{thm:tree_packing} also yields a tree packing of at least $\lambda$ spanning trees with diameter $O((n\log n)/\delta)$ where each edge belongs to $O(\log n)$ trees. There is an alternative proof to show the existence of such a tree packing using the techniques in~\cite{DBLP:conf/icalp/ChuzhoyPT20}. The proof was suggested by an anonymous reviewer, who kindly allowed us to include the proof in our paper, see \cref{sec:existential}.

\paragraph{Optimality.}
Due to an $\OPT = \Omega(k / \lambda)$ information-theoretic lower bound, in the regime where $k = \Omega(n)$, our broadcast algorithm is \emph{universally optimal}, up to a logarithmic factor, as its round complexity is $O(\OPT \log n)$ for \emph{any graph} $G$, where $\OPT$ is the round complexity of an optimal broadcast algorithm on $G$ that is designed especially for $G$ and knows the entire topology of $G$. 
Our broadcast algorithm implies that $\OPT$ is within an $O(\log n)$ factor of $k/\lambda$  when $k = \Omega(n)$, addressing \cref{g2}. Moreover, in the regime where $k = \Omega(n)$, the decomposition of \cref{thm:tree_packing} implies that a fractional tree packing with \emph{the same} parameters as that of Ghaffari~\cite{DBLP:conf/icalp/Ghaffari15} can be constructed in $O(\OPT \log n)$ rounds, addressing \cref{g1}. Furthermore, our approach only uses integral weights.
In addition, for all values of $k \leq n$, our algorithm nearly matches an \emph{existential} lower bound of Ghaffari and Kuhn~\cite{DBLP:conf/wdag/GhaffariK13}, see \cref{sec:broadcast} for the details.

\paragraph{Remark.} We emphasize that the algorithm of \cref{thrm:main}, which is based on \cref{thm:tree_packing}, requires the \emph{knowledge} of edge connectivity $\lambda$. As we will later discuss, the value of $\lambda$ can be learned in $\tilde{O}(n/\delta)$ rounds using techniques from \cite{DBLP:conf/icalp/ChuzhoyPT20,ghaffari2022universally}, so \kbroadcast can be still solved \whp in $\tilde{O}((n+k)/\lambda)$ rounds without the knowledge of $\lambda$. There is an alternative approach by doing an exponential search: Compute the decomposition of \cref{thm:tree_packing} with $\tilde{\lambda} = \delta, \delta/2, \delta/4, \ldots$ until it yields a desired tree packing. The total number of iterations is $O\left(\log \frac{\delta}{\lambda} \right)$. Checking the validity of a tree packing takes $O((n\log n)/\delta)$ rounds, as we just need to verify whether each $G_i$ is a connected subgraph with diameter $O((n\log n)/\delta)$. Thus, the overall cost can be upper bounded by $O((n\log n)/\lambda)$. Based on this approach, we infer that \kbroadcast can be solved \whp in $O(((n+k)/\lambda)\log n)$ rounds without the knowledge of $\lambda$.

\subsection{Applications}

As an important special case of \cref{thrm:main}, we can broadcast $k = \Theta(n)$ messages in $O((n \log n)/\lambda)$ rounds.
In particular, in this round complexity, each node can broadcast a message of $O(\log n)$ bits to all other nodes in the graph. This immediately yields a simulation of one round of the \emph{broadcast congested clique} model~\cite{drucker2014on}. The round complexity $O((n \log n)/\lambda)$ of the simulation is universally optimal up to a logarithmic factor.
 Our broadcast algorithm can also be used to broadcast a spanner or a sparsifier to the whole graph efficiently, leading to new approximation algorithms for distances and cuts.

\subsubsection*{Approximate all pairs shortest paths}
Using our broadcast algorithm, we obtain fast distributed algorithms for approximate \emph{All Pairs Shortest Paths} (APSP).
We say that a distance estimate $\tilde{d}$ is an $(\alpha, \beta)$-\emph{approximation} of the true shortest-path distance $d$ if \[d(u,v) \leq \tilde{d}(u,v) \leq \alpha d(u,v) + \beta\] for all nodes $u$ and $v$ in the graph. In other words, $\alpha$ is the multiplicative error and $\beta$ is the additive error. If there is no additive error, we write $\alpha$-\emph{approximation} to denote $(\alpha, 0)$-approximation.
We show randomized algorithms with the following guarantees.

\begin{enumerate}
    \item $(3, 2)$-approximate unweighted APSP takes  $\tilde{O}(n/\lambda)$ rounds.\label{unweigh}
    \item $(2k-1)$-approximate weighted APSP takes $\tilde{O}\left(n^{1+\frac{1}{k}}/\lambda\right)$ rounds, for any integer $k \geq 1$.\label{weig_k}
    \item $O\left(\frac{\log n}{\log \log n}\right)$-approximate weighted APSP takes  $\tilde{O}(n/\lambda)$ rounds. \label{weig_log}
\end{enumerate}

\cref{unweigh} is obtained using a clustering approach. First, we show that a graph with minimum degree $\delta$ can be decomposed into $\tilde{O}(n/\delta)$ clusters of constant diameter. Then we exploit the fact that the cluster graph only has $\tilde{O}(n/\delta)$ clusters to simulate a linear-time algorithm for unweighted APSP by
Peleg, Roditty, and Tal \cite{unweightedapsp} in just $\tilde{O}(n/\delta)$ rounds on the cluster graph. To obtain estimates of the distances in the \emph{original} graph all nodes should know to which cluster each node belongs. This is a task of broadcasting $n$ messages, which can be solved efficiently by our broadcast algorithm. See  \cref{sec:APSP} for details.

\cref{weig_k} is obtained by computing a \emph{spanner}, a sparse subgraph that estimates the distances, and broadcasting the spanner to the whole graph. \cref{weig_log} is a special case of \cref{weig_k}.
Here we use the Baswana--Sen algorithm~\cite{weightedspanner}, which constructs a spanner with $\tilde{m} = O\left(k \cdot n^{1+\frac{1}{k}}\right)$ edges that preserves the distances up to a multiplicative $(2k-1)$ factor in $O(k^2)$ rounds, and then we use
  \cref{thrm:main} to broadcast the spanner to the entire graph.
See \cref{sec:wAPSP} for the details. 

\paragraph{Optimality.}
Our APSP algorithms are \emph{universally optimal} in the following sense: In order to write down the estimates of all distances or all cut sizes, it is \emph{necessary} to first learn the list of all IDs in the graph, and there is a simple information-theoretic  $\Omega(n / \lambda)$ universal lower bound for learning the list of all IDs which holds for any graph. Therefore, any algorithm that solves approximate APSP in $o(n / \lambda)$ rounds for certain graphs must involve strange tricks that recompute the IDs of some nodes in these graphs or must work under the assumption that the list of all IDs was known to the algorithm.

In the weighted case, we show that $\Omega(n / \lambda)$ is still a lower bound for approximate APSP even if the IDs are initially known to all nodes. This is an \emph{existential} lower bound in the sense that the lower bound only applies to a special family of graphs. The proof idea is to encode messages as edge weights in such a way that solving approximate APSP requires some node to learn all the messages precisely. See Section \ref{sec:APSP_LB} for the details.

 To the best of our knowledge, our results are the first \emph{sublinear} algorithms for APSP for graphs with high edge connectivity. APSP is a central and well-studied problem. In general graphs, there are $\tilde{O}(n)$-round \congest algorithms \cite{nanongkai2014distributed,DBLP:conf/podc/LenzenP15,lenzen2013fast,unweightedapsp,holzer2012optimal,DBLP:conf/stoc/BernsteinN19} even for exact APSP in weighted graphs \cite{DBLP:conf/stoc/BernsteinN19}, which is tight due to an $\tilde{\Omega}(n)$ lower bound \cite{nanongkai2014distributed,DBLP:conf/wdag/Censor-HillelKP17,frischknecht2012networks,holzer2012optimal,DBLP:conf/wdag/AbboudCK16}. The lower bound from \cite{nanongkai2014distributed} holds even for any polynomial approximation for weighted graphs or any polylogarithmic approximation for unweighted graphs. This lower bound is shown in a graph family where $\lambda=1$, and we prove that for graphs with higher edge connectivity, we can get much faster algorithms. We remark that we cannot hope for a similar result for exact APSP, and a certain approximation is needed, as \cite{frischknecht2012networks,holzer2012optimal} show a family of graphs with edge connectivity $\lambda = \Omega(n)$ where obtaining a $(3/2-\epsilon)$-approximation for unweighted APSP takes $\Omega(n/\log{n})$ rounds. 

 \paragraph{Comparison with prior work.}
In a recent work~\cite{DBLP:conf/spaa/Censor-HillelLP21}, Censor-Hillel,  Leitersdorf, and  Polosukhin designed \emph{sublinear} algorithms for APSP in graphs with small \emph{mixing time}.
They showed that a $(3+\epsilon)$ approximation for weighted, undirected APSP can be obtained in $(n^{1/2}+n/\delta) \cdot \tmix \cdot 2^{O(\sqrt{\log{n}})}$ rounds \whp in \congest, where $\tmix$ is the mixing time of the graph.  Their algorithm is based on developing efficient simulations of algorithms in the congested clique that exploits the assumption that the minimum degree $\delta$ and the mixing time $\tmix$ are high.

Our new APSP algorithms \emph{improve} over their algorithms: The round complexity $\tilde{O}(n/\lambda)$ of our APSP algorithms is better than their round complexity $(n^{1/2}+n/\delta) \cdot \tmix \cdot 2^{O(\sqrt{\log{n}})}$ for \emph{any graph}, as we always have $\tmix = \Omega(\delta / \lambda)$. To see this, recall that $\tmix = \Omega(1/\phi)$, where $\phi$ is the conductance of the graph, and we observe that $\phi = O(\lambda / \delta)$ by considering an arbitrary minimum cut. It is worth noting that our approach is \emph{significantly simpler} than the approach taken in~\cite{DBLP:conf/spaa/Censor-HillelLP21}. 

\subsubsection*{Approximating cuts}
An additional application of our broadcast algorithm is that it allows us to distribute a \emph{spectral sparsifier} to the whole graph efficiently, which in particular allows us to approximate all the cut sizes in the graph. Using the spectral sparsifier of Koutis and Xu~\cite{koutis2016simple}, we can approximate the weights of all the cuts in the graph within a $(1+\epsilon)$ factor in $\tilde{O}(n/(\lambda \epsilon^2))$ rounds \whp. See Section \ref{sec:cuts} for the details. 

There are previous algorithms that approximate the minimum cut~\cite{DBLP:conf/icalp/ChuzhoyPT20,ghaffari2020faster} or the maximum cut~\cite{censor2017fast,DBLP:conf/wdag/KawarabayashiS18} faster. While $(1+\epsilon)$-approximate weighted minimum cut can already be computed in $\tilde{O}(n/\lambda)$ rounds~\cite{DBLP:conf/icalp/ChuzhoyPT20}, building a spectral sparsifier allows us to approximate \emph{all the cuts} in the graph simultaneously. To the best of our knowledge, our algorithm is the first one to achieve this goal in \emph{sublinear}-in-$n$ rounds. 

\subsubsection*{An application to secure distributed computing}  The tree packing algorithm resulting from \cref{thm:tree_packing} is of independent interest due to the wide applications of tree packings. In particular, it has an application to resilient and secure distributed graph algorithms, which is a topic that has received a lot of attention recently~\cite{DBLP:conf/podc/ParterY19,DBLP:conf/soda/ParterY19a,DBLP:conf/wdag/HitronP21a,DBLP:conf/wdag/HitronP21,DBLP:conf/wdag/HitronPY22,DBLP:conf/podc/Parter22,DBLP:conf/podc/FischerP23,DBLP:conf/innovations/HitronPY23}.
 In these works, compilation schemes that turn any given distributed algorithm into a resilient and secure one were designed.
 An algorithm is said to be \emph{$f$-mobile-resilient} if the algorithm works correctly even with the presence of an adversary that can adaptively control a possibly different subset of $f$ edges in each round. 
In a recent work by Fischer and Parter~\cite{DBLP:conf/podc/FischerP23}, they showed that if a  tree packing with at least $\lambda$ trees with congestion $\polylog n$ and tree diameter $d$ is given, then there is a compiler that turns any given \congest algorithm into an $f$-mobile-resilient one with $f = \tilde{O}(\lambda)$ such that the round complexity overhead of the compilation is $\tilde{O}(d)$.
 To complement this result, Fischer and Parter showed a tree packing algorithm that achieves a near-optimal value of $d$, up to a polylogarithmic factor, in  $\tilde{O}(\lambda d^2)$ rounds in the \congest model using the approach of Ghaffari~\cite{DBLP:conf/icalp/Ghaffari15}.
 \cref{thm:tree_packing} can be applied to construct a tree packing with $d = O((n\log n)/\delta)$ that is suitable for their application. In the setting where $\lambda$ and the optimal value of $d$ are large, our approach leads to a faster algorithm.

\subsection{Additional Related Work}\label{sec:related}

Due to the wide applications of broadcast and tree packing, our work naturally connects with various research topics in distributed computing.  In this section, we examine related research on these topics.

\paragraph{Universally optimal algorithms via low-congestion shortcuts.} An algorithm is  \emph{universal optimal} if its complexity on any instance is nearly the complexity of the best algorithm that is designed specifically for this instance. The concept of universal optimality was introduced in~\cite{Garay1998} and has attracted a lot of attention in recent years in the distributed setting. 

A fruitful line of research, based on the notion of \emph{low-congestion shortcuts}~\cite{ghaffari2016distributed} led to universally optimal algorithms for many central problems such as minimum spanning tree~\cite{ghaffari2016distributed}, minimum cut~\cite{ghaffari2016distributed, ghaffari2022universally}, and approximate single source shortest paths~\cite{DBLP:conf/soda/ZuzicGYHS22,DBLP:conf/stoc/RozhonGHZL22}.

Earlier works in this series of research focus on special graph classes such as  planar graphs~\cite{ghaffari2016distributed}, bounded-genus graphs~\cite{haeupler2016near}, or even any $H$-minor-free graphs~\cite{haeupler2018minor,ghaffari2021low}. More recent works in this direction were able to achieve nearly universally optimal complexity for some problems in general graphs assuming that the graph topology is known to the algorithm~\cite{universally_optimal_stoc2021}.

For some problems in general graphs, including minimum spanning tree, $(1+\epsilon)$-approximate single-source shortest paths, and $(1+\epsilon)$-approximate minimum cut, it is possible to achieve a round complexity that is polynomial in $\OPT$ without the known topology assumption, where $\OPT$ is the complexity of the best algorithm that is designed specifically for the underlying graph~\cite[Theorem 1.4]{haeupler2022hop}.

A limitation of the low-congestion shortcut framework is that it is only suitable to tackle a certain class of problems, such as minimum spanning tree, minimum cut, and approximate single source shortest paths. All these problems can be solved in $\tilde{O}(D+\sqrt{n})$ rounds in general graphs. For these problems, it is known that shortcuts capture their optimal round complexities very well~\cite{universally_optimal_stoc2021}. 

Our work, together with the prior work by Ghaffari~\cite{DBLP:conf/icalp/Ghaffari15}, offers a new approach for designing universal optimal algorithms for problems that cannot be solved using the low-congestion shortcut framework, such as broadcast and approximate APSP. We believe that our approach will find further applications for characterizing the universally optimal complexity as well as designing universally optimal algorithms for other problems in the future.

\paragraph{Distributed algorithms in highly connected graphs.}
Chuzhoy, Parter, and Tan also utilized their distributed tree packing algorithm to obtain improved algorithms in highly connected graphs~\cite{DBLP:conf/icalp/ChuzhoyPT20}.
 They showed that a minimum spanning tree and a minimum cut can be computed in $\tilde{O} (\min \{ \sqrt{n/\lambda} + n^{D/(2D+1)}, n/\lambda\})$ rounds, and $n^{o(1)}$-approximate
  single-source shortest paths can be computed in ${O} (\min \{ \sqrt{n/\lambda} + n^{D/(2D+1)}, n/\lambda\}) \cdot n^{o(1)}$ rounds. In particular, all these problems can be solved in $\tilde{O}(n/\lambda)$ rounds. For unweighted minimum cut, it might be even possible to get a better complexity of $\tilde{O}(D+\sqrt{n/\lambda})$ rounds~\cite[Footnote 4]{ghaffari2020faster}. 
  
  These problems are quite different in nature compared to the problems that we tackle in this work in that these problems can be solved in $\tilde{O}(D+\sqrt{n})$ rounds in general graphs, which is already sublinear in $n$. In contrast, as discussed earlier, approximate APSP requires $\tilde{\Omega}(n)$ rounds to solve in general graphs.

  At a high level, the problems considered in~\cite{DBLP:conf/icalp/ChuzhoyPT20} can be solved by breaking the computation into a series of aggregate computations. For example, the minimum spanning tree problem can be solved in $O(\log{n})$ iterations where in each iteration the goal is to find minimum outgoing edges from connected components. Aggregation tasks such as computing the minimum values in disjoint connected components can be solved efficiently in $\tilde{O}(n/\lambda)$ rounds as shown in \cite{DBLP:conf/icalp/ChuzhoyPT20}. On the other hand, when we solve broadcast it is not enough to compute an aggregate function of the messages as we need to deliver all of them.

In addition to information dissemination, there exist distributed algorithms that exploit the high connectivity of the underlying graph to deal with failures and to design secure distributed algorithms~\cite{censor2018fast, DBLP:conf/icalp/ChuzhoyPT20, DBLP:conf/wdag/HitronPY22}. There is also a line of research that studies distributed algorithms that depend on the \emph{vertex connectivity} of the underlying graph~\cite{DBLP:conf/podc/Censor-HillelGK14,DBLP:journals/talg/Censor-HillelGG17, DBLP:conf/soda/Censor-HillelGK14}. 


\paragraph{Additional related works on cuts and connectivity.} There is a large body of research related to cuts and connectivity in distributed networks. This includes distributed algorithms for computing minimum cuts \cite{pritchard2011fast, DBLP:conf/wdag/GhaffariK13, DBLP:conf/wdag/NanongkaiS14, DBLP:conf/stoc/DagaHNS19, DBLP:conf/wdag/Parter19, ghaffari2020faster,DBLP:conf/stoc/DoryEMN21, ghaffari2022universally}, minimum vertex cuts \cite{pritchard2011fast, DBLP:conf/podc/Censor-HillelGK14, DBLP:conf/wdag/Parter19, DBLP:conf/wdag/ParterP22,jiang2023finding}, connectivity certificates \cite{thurimella1995sub, DBLP:conf/wdag/Parter19,DBLP:conf/stoc/DagaHNS19, bezdrighin2022deterministic}, and low-cost $k$-edge-connected graphs \cite{DBLP:journals/dc/Censor-HillelD20, DBLP:conf/podc/Dory18, DBLP:conf/podc/DoryG19, DBLP:conf/soda/Dory023}. 

In particular, tree packings play an important role in designing fast algorithms for computing a minimum cut in a graph. Tree packings were first used in the seminal near-linear time centralized algorithm of Karger~\cite{karger2000minimum} and was later used in many algorithms, including the recent centralized algorithms in~\cite{DBLP:conf/stoc/MukhopadhyayN20,DBLP:conf/icalp/GawrychowskiMW20,DBLP:conf/sosa/GawrychowskiMW21}. In the \congest model, this approach based on tree packings led to near-optimal distributed algorithms for the minimum cut problem~\cite{DBLP:conf/wdag/NanongkaiS14,DBLP:conf/stoc/DagaHNS19,ghaffari2020faster, DBLP:conf/stoc/DoryEMN21, ghaffari2022universally}. In these algorithms, usually, the goal is to find a small collection of spanning trees such that one of them has at most two edges of the minimum cut, and then this property can be exploited to find the minimum cut of the graph. Unlike the applications in information dissemination, the diameter of the trees in the tree packing is not relevant to the applications in minimum cut computation. This is because long-distance communication in a connected subgraph is facilitated by low-congestion shortcuts in these applications.

\subsection{Roadmap}
In \cref{sec:prelim}, we present the basic definitions and algorithmic tools.
In \cref{sect:broadcast}, we present our broadcast algorithm.
In \cref{sec:APSP_all}, we present applications of our broadcast algorithms for approximating distances and cuts in graphs. 
In \cref{sec:existential}, we discuss an alternative approach for low-diameter tree packings. 
In \cref{appendix_lower_bound}, we show a lower bound for tree packings with low congestion.

\section{Preliminaries}\label{sec:prelim}

Throughout the paper, for a graph $G = (V, E)$, we write $n = |V|$, $m = |E|$, $D = \diam(G)$, $\delta =$ the minimum degree of $G$, and $\lambda =$ the edge connectivity of $G$. For our algorithms we assume that the values of $\delta$ and $\lambda$ are known by all nodes, as we show that we can learn these values without increasing the asymptotic round complexity of the algorithms. 
For each node $v \in V$, we write $\deg(v)$ to denote the number of edges incident to $v$ and write $N(v)$ to denote the set of neighbors of $v$.  For a subset of nodes $A \subseteq V$, we write $\deg_{A}(v)$ to denote the number of edges between $v$ and the nodes in $A$.
 We use $\tilde{O}(\cdot)$, $\tilde{\Theta}(\cdot)$, and $\tilde{\Omega}(\cdot)$ to hide any $\poly \log n$ factors. We say that an event occurs \emph{with high probability} (\whp) if the event occurs with probability $1 - 1/\poly(n)$. For any positive integer $x$, we write $[x] = \{1, 2, \ldots, x\}$.

Unless otherwise stated, we assume that the graph $G$ under consideration is \emph{simple}, \emph{connected}, \emph{unweighted}, and \emph{undirected}, so we always have $\delta \geq \lambda \geq 1$ and $D = O(n/\delta) = O(n/\lambda)$.

\begin{observation}[Diameter upper bound]
\label{lemma:diameter_vs_lambda}
    The diameter $D$ of any simple graph is $O(n/\delta)$.
\end{observation}
\begin{proof}
    Let there be any shortest path $P = (v_1, \dots, v_{D+1})$ of length $D$. For any $i \in [D-2]$, observe that $N(v_i) \cap N(v_{i+3}) = \emptyset$. To see this, assume that some $u \in N(v_i) \cap N(v_{i+3})$, then, there is a path of length $2$ between $v_i$ and $v_{i+3}$. However, the path between $v_i$ and $v_{i+3}$ in $P$ is of length $3$, contradicting the fact that $P$ is a shortest path.
    Therefore, as $N(v_i) \cap N(v_{i+3}) = \emptyset$, we conclude that $n \geq \sum_{j \in [D/3]} |N(v_{3j - 2})| \geq (D/3)\cdot \delta$, implying $D = O(n/\delta)$.
\end{proof}

\paragraph{Model of distributed computing.}
We consider the \congest model of distributed computing~\cite{peleg2000distributed}, where the network is abstracted as an undirected graph $G=(V,E)$ where each node $v \in V$ corresponds to a computer and each edge $e \in\{u,v\}$ corresponds to a communication link. We assume each node $v \in V$ has a unique $O(\log n)$-bit identifier $\ID(v)$, that is, $\ID(v) \in [n^c]$ for some constant $c \geq 1$. The communication proceeds in synchronous rounds, where per round each node can send one $O(\log n)$-bit message to each of its neighbors. 
Initially, the topology of the network $G$ is unknown to the nodes in the network. At the end of the computation, each node is required to know its part of the output. For instance, for the All-Pairs Shortest Paths (APSP) problem, each node needs to learn its distance to every other node in the graph.

\paragraph{The broadcast problem.} 
In the subsequent discussion, we discuss the basic result of the \kbroadcast problem.
We define the \emph{congestion} of an algorithm as the maximum number of messages passing through an edge $e$ over the whole execution of the algorithm, ranging over all edges $e \in E$ in the graph.
It is well-known that once a leader is elected, the \kbroadcast problem can be solved in $O(D + k)$ rounds with congestion $O(k)$. Roughly, the algorithm first computes a BFS tree $T$ rooted at the leader in $O(D)$ rounds. After that, the messages can be broadcast to all nodes in $O(k + D)$ rounds in a pipelined fashion in $T$. This algorithm is useful for small $k$ in graphs with a small diameter.

\begin{lemma}[Basic broadcast algorithm \cite{peleg2000distributed}]
\label{very_old_k_broadcast}
     Given that there is a unique leader in the graph $G=(V,E)$, \kbroadcast  can be solved in $O(k + D)$ rounds deterministically in such a way that the number of messages passing through an edge $e$ is at most $O(k)$, for each $e\in E$.
\end{lemma}

\paragraph{BFS, leader election, and ID assignment.} It is well-known that a breadth-first search tree (BFS) can be computed in $O(D)$ rounds in \congest deterministically, without requiring any knowledge of $n$ or $D$. The BFS algorithm works as long as each node is given a unique identifier that fits in  $O(\log n)$ bits.

\begin{lemma}
[Breadth-first search~\cite{peleg2000distributed}]
\label{lemma:construct_bfs_tree}
    Given any graph $G$, it is possible to compute a BFS tree $T$ in $G$ in $O(D)$ rounds in such a way that each node knows which of its incident edges are in $T$.  The algorithm works as long as each node is given a unique identifier that fits in  $O(\log n)$ bits and does not require any knowledge about the graph.
\end{lemma}

BFS can be used to solve the \emph{leader election} problem by setting the root as the leader. BFS can also be used to compute new identifiers. 

 \begin{restatable}{lemma}{IDassign} \label{lemma:number_items}
    Suppose each node $v$ originally holds a number $x_v \leq n^c$ of items, for some constant $c > 0$. Let $X = \sum_{v \in V} x_v$. It is possible to assign distinct identifiers to the items in the graph $G$ such that each item has a unique identifier in $[X]$ and each node $v$ knows the identifiers of its items. The process completes in $O(D)$ rounds.
 \end{restatable}

 \begin{proof}
    Compute a BFS tree $T$ of $G$ using \cref{lemma:construct_bfs_tree} in $O(D)$ rounds. Let $v_0$ be the root of $T$.   
    Over $T$, each parent sums the number of items in its subtrees. After $O(D)$ rounds, $v_0$ knows the number of items in each of its subtrees. Node $v_0$ takes identifiers $[x_{v_0}]$ for its own items, and subdivides the range of identifiers $\{x_{v_0} + 1, x_{v_0} + 2, \ldots, X\}$ to each of its children, according to the number of items in each of its subtrees. In turn, each of its children takes identifiers for its items and subdivides the range of identifiers it received from its parent to its children. After $O(D)$ rounds, each item has a unique identifier in $[X]$.
\end{proof}

\paragraph{Knowledge of $\delta$ and $\lambda$.}

We show that all nodes can learn the values of $\delta$ and $\lambda$ in $\tilde{O}(n/\delta)$ rounds. 

\begin{lemma}
    Given a graph $G=(V,E)$, all nodes can learn the values of $\delta$ and $\lambda$ in $\tilde{O}(n/\delta)$ rounds.
\end{lemma}

\begin{proof}
    Learning the value of $\delta$, the minimum degree in the graph, can be done in $O(D)=O(n/\delta)$ rounds via computation of a minimum over a BFS tree, as initially each node knows its own degree. The value of $\delta$ can be broadcast to all nodes in additional $O(D)$ rounds.  

  The value of $\lambda$ can be learned in $\tilde{O}(n/\delta)$ rounds using techniques from \cite{DBLP:conf/icalp/ChuzhoyPT20,ghaffari2022universally}. The work \cite{ghaffari2022universally} shows that the value of the minimum cut, which is $\lambda$ is unweighted graphs, can be computed in $\tilde{O}(T_{\mathcal{SQ}}+\mathcal{SQ}(G))$ time, where $\mathcal{SQ}(G)$ is a parameter called the shortcut quality of the graph, and $T_{\mathcal{SQ}}$ is the running time for constructing shortcuts with quality $\mathcal{SQ}(G)$. 
  The work \cite{DBLP:conf/icalp/ChuzhoyPT20} shows that shortcuts with quality $O(n/\lambda)$ can be constructed in $O(1)$ rounds in graphs with edge connectivity $\lambda$ (Theorem 7.10). A closer look at their proof shows that they only use the fact that the minimum degree in the graph is at least $\lambda$. Hence, the exact same proof shows that shortcuts with quality $O(n/\delta)$ can be constructed in $O(1)$ rounds in graphs with edge connectivity $\lambda$. Combining \cite{DBLP:conf/icalp/ChuzhoyPT20, ghaffari2022universally} gives an algorithm with running time $\tilde{O}(n/\delta)$ for computing 
  $\lambda$. 
\end{proof}

If we aim for algorithms that take $\tilde{O}(n/\lambda)$ rounds, we can assume that all nodes initially know the values of $\delta$ and $\lambda$. As we will later see, the knowledge of $\lambda$ is not necessary for our broadcast algorithm to work, as we can use an exponential search to guess the value of $\lambda$.

\section{The Broadcast Algorithm}\label{sect:broadcast}

We first show a key lemma that random edge sampling with probability $p=C\log n/\lambda$ for a sufficiently large constant $C$ yields a spanning subgraph with diameter $O(n\log n/\delta)$.

\begin{lemma}\label{lem-mainlemma}
    Let $C = \Omega(1)$. Let $G=(V, E)$ be an $n$-node simple graph with edge connectivity $\lambda$ and minimum degree $\delta$. Let $G'=(V, E')$ be a  subgraph of $G$ where each edge $e\in E$ is included in $E'$ with probability $p=C\log n/\lambda$ independently, then $G'$ is a spanning subgraph with diameter $O(C n\log n/\delta)$ with probability $1 - n^{-\Omega(C)}$.
\end{lemma}
\begin{proof}
We view the $p$-probability sampling as $L=\Theta(C \log n)$ independent iterations of $q$-probability sampling for $q=1/\lambda$, where an edge is called sampled in the former if it is sampled in any of the iterations of the latter. Precisely, the relation is that $1-p = (1-q)^{L}$.

We reveal the edges sampled in the $L$ iterations and
argue that for each node $v$, its $L$-hop neighborhood through these edges has at least $\delta/4$ nodes \whp. We reveal these edges in $L$ iterations. For $j\in[L]$, let $B_{j}(v)$ be the set of nodes reached from $v$ up to distance $j$ using the edges in the first $j$ iterations of sampling. Next, we determine $B_{j+1}(v)$. For that, we examine the nodes in $B_{j}(v)$ one by one and consider their edges going outside the currently known $B_{j+1}(v)$. We say \textit{node $u$ fails in step $j+1$} if, when we reveal the sampling of this iteration for edges incident on $u$, there is no sampled edge in iteration $j+1$ that connects to a \textit{new node}, i.e., a node outside the currently known $B_{j+1}(v)$. If ever during the process, the current size of $B_{j+1}(v)$ is above $\delta/4$, we are done and know that we will have $|B_{L}(v)|\geq \delta/4$. Otherwise, each node $u\in B_{j}(v)$ fails with probability at most $(1-1/\lambda)^{3\delta/4}<1/2$, where the last inequality uses that $\delta\geq \lambda$. The probability bound holds because $u$ has at least $3\delta/4$ edges to nodes outside the current $B_{j+1}(v)$,\footnote{This part crucially assumes that $G$ is a simple graph; it would break for a multigraph.} and each of these is sampled with probability at least $1/\lambda$ in this iteration. Let us call the iteration an \textit{overall failure} if at least $2|B_j(v)|/3$ nodes fail. The probability of an overall failure is at most $3/4$, since otherwise the expected number of failed nodes in $B_j(v)$ exceeds $|B_{j}(v)|/2$. Moreover, if we are not in an overall failure situation, we have $|B_{j+1}(v)| \geq \min\{\delta/4, 4|B_{j}(v)|/3\}$. Since we always have $|B_{j+1}(v)| \geq |B_{j}(v)|$, given that $L=\Theta(C \log n)$, we conclude that $|B_{L}(v)|\geq \delta/4$ for all $v \in V$ with probability $1 - n^{-\Omega(C)}$ by a Chernoff bound for each $v \in V$ and a union bound over all $v \in V$.

We now use the above result to conclude that $\diam(G')\leq 20nL/\delta = \Theta(C n\log n/\delta)$ \whp by an argument similar to the proof of \cref{lemma:diameter_vs_lambda}. By Karger's result \cite{karger1999random} 
applied to $G'$, we know that $G'$ is connected with probability $1 - n^{-\Omega(C)}$.  Suppose for the sake of contradiction that $\diam(G')> 20nL/\delta$ and choose $v \in V$ and $u \in V$ such that $\dist_{G'}(v, u) > 20nL/\delta$. Let $P_{v,u}=(w_0, w_1, w_2, \dots, w_\ell)$ be a shortest path in $G'$ between $w_0=v$ and $w_\ell=u$, where $\ell=\dist_{G'}(v, u) > 20nL/\delta$. Let $W=\{w_{0}, w_{3L}, w_{6L}, \dots, w_{3L\cdot\lfloor{\ell/3L\rfloor}}\}$, and observe that $|W|\geq \lfloor \ell/3L\rfloor \geq 5n/\delta$.  Since $P_{v, u}$ is a shortest path, we know that the $L$-hop neighborhoods $B_{L}(w_k)$ of nodes $w_k\in W$, as defined above, must be disjoint; any intersection would imply a shorter path. From above we know that $|B_{L}(w_k)|\geq \delta/4$ for each $w_k\in W$. This implies that the number of nodes is at least $|W| \cdot \delta/4 = 5n/\delta \cdot \delta/4 =5n/4 > n$, which is a contradiction. Having arrived at this contradiction from the assumption that $\diam(G')> 20nL/\delta$, we conclude that $\diam(G')\leq 20nL/\delta = O(C n\log n/\delta)$.
\end{proof}

 Our key lemma can be seen as a strengthening of Karger's well-known connectivity under random edge sampling result~\cite{karger1999random}. Karger shows that if the sampling probability is $p=\Omega(\log n/\lambda)$, then the sample graph is connected \whp and thus has diameter at most $n-1$. We show that under the same condition, the diameter is indeed $\tilde{O}(n/\lambda)$ \whp.

\subsection{Low-Diameter Tree Packings}

\cref{thm:tree_packing} follows from \cref{lem-mainlemma} immediately.

\treepackinga*
\begin{proof}
The theorem follows from applying \cref{lem-mainlemma} to $G_i$  for all $i \in [\lambda']$. By a union bound over all $i \in [\lambda']$, the success probability is $1 - \lambda \cdot n^{-\Omega(C)} = 1 - n^{-\Omega(C)}$.
\end{proof}

\paragraph{Tree packings.}
By spending $O((n\log n)/\delta)$ rounds to do a BFS (\cref{lemma:construct_bfs_tree}) in parallel for all the edge-disjoint spanning subgraphs in \cref{thm:tree_packing}, we may obtain a \emph{tree packing} of $\Omega(\lambda / \log n)$ \emph{edge-disjoint} spanning trees with diameter $O((n\log n)/\delta)$.   

Such a tree packing can be viewed as a \emph{fractional} tree packing with total weight $\Omega(\lambda / \log n)$ by assigning the same unit weight to all trees in the tree packing.
In the regime of $k = \Omega(n)$, we can rewrite $\Omega(\lambda / \log n)$ as $\Omega\left(k / (\OPT \log n)\right)$, where $\OPT = \Omega(k/\lambda)$ is the optimal round complexity for broadcasting $k$ messages of the underlying graph $G$ when the topology of $G$ is known to the algorithm. 

By \cref{thrm:main,thm:universalLB}, we know that $\OPT$ is within an $O(\log n)$ factor of $k/\lambda$ in the regime of $k = \Omega(n)$, so the diameter of the tree packing $O((n\log n)/\delta)$ can be upper bounded by $O(\OPT \log n)$. Therefore, in the context of fractional tree packing, when $k = \Omega(n)$, \cref{thm:tree_packing} implies that a fractional tree packing with total weight $\Omega\left(k / (\OPT \log n)\right)$ and diameter $O(\OPT \log n)$ can be computed in $O(\OPT \log n)$ rounds. That is, in the regime of $k = \Omega(n)$, we show that a fractional tree packing with exactly the same parameters as that of Ghaffari~\cite{DBLP:conf/icalp/Ghaffari15} can be constructed in $O(\OPT \log n)$ rounds, addressing \cref{g1}. Furthermore, our approach only uses integral weights.

\paragraph{Lower bounds.} 
The parameters of our tree packings nearly match the \emph{existential} lower bounds from \cite{DBLP:conf/wdag/GhaffariK13}, which showed a family of graphs with $n$ nodes and diameter $O(\log{n})$ where in any tree packing the diameter of all trees is $\Omega(n/\lambda)$, except at most $O(\log{n})$ trees that may have a smaller diameter. Therefore, the diameter bound $O((n\log n)/\delta)$ in \cref{thm:tree_packing} is \emph{optimal} up to an $O(\log n)$ factor.
Their lower bound is for the case where all the trees are edge-disjoint, but the lower bound can also be extended to the more general case where the congestion is $O(\lambda/\log^4{n})$, i.e., each edge belongs to $O(\lambda/\log^4{n})$ trees. See \cref{appendix_lower_bound} for details.

\subsection{Broadcast} \label{sec:broadcast}

Our main contribution in this work is a near-optimal algorithm for the broadcast problem.

\broadcast*
\begin{proof}
Compute $\lambda' = \Omega(\lambda / \log n)$ edge-disjoint spanning subgraphs $G_1=(V, E_1)$, \ldots, $G_{\lambda'}=(V, E_{\lambda'})$ by \cref{thm:tree_packing}.
 Number the messages $1, 2, \dots, k$ via \cref{lemma:number_items} in $O(D) = O(n/\delta)$ rounds (\cref{lemma:diameter_vs_lambda}), and assign messages with numbers in $[(i-1)K+1, iK]$ to subgraph $G_i$, where $K=\lceil k/\lambda' \rceil$. Each subgraph $G_i$ gets $k_i=O(k/\lambda')=O((k\log n)/\lambda)$ messages assigned to it. Using \cref{very_old_k_broadcast}, we may broadcast these messages in $G_i$ in $O(\diam(G_i)+k_i)= O((n\log n)/\delta+(k\log n)/\lambda)$ rounds, in parallel for all $i \in [\lambda']$.
\end{proof}

\paragraph{Existential optimality.}
For all values of $k \leq n$, our algorithm nearly matches the \emph{existential} lower bound of Ghaffari and Kuhn~\cite{DBLP:conf/wdag/GhaffariK13} (\cref{lower_bound_broadcast}) which constructs a family of graphs where \[\Omega\left(D + \min\left\{ \frac{K}{\log^2{n}}, \frac{n}{\lambda} \right\}\right)\] rounds are needed to solve the easier \emph{unicast} problem, where a node $s$ should send $K$ bits of information to a node $t$. This implies a lower bound also for the harder \kbroadcast problem with $k = K/\log n$, where we should send $K=O(k \log{n})$ bits of information to all the nodes in the graph. For all values of $k$, we can nearly match this lower bound by combining \cref{thrm:main} with the $O(D+k)$-round textbook broadcast algorithm (\cref{very_old_k_broadcast}). By doing so, we infer that \kbroadcast  can be solved in \[\min\left\{O(D+k), O\left(\frac{n\log n}{\delta}+\frac{k\log n}{\lambda}\right)\right\} \ \ \text{rounds.}\]

\paragraph{Universal lower bound.}
We next show that in the regime where $k = \Omega(n)$, our broadcast algorithm is  \emph{universally optimal} in the sense that its round complexity is $O(((n+k)/\lambda)\log n) = O((k/\lambda) \log n) = O(\OPT \log n)$ for \emph{any graph} $G$, where $\OPT = \Omega(k / \lambda)$ is the round complexity of an optimal broadcast algorithm on $G$ that is designed especially for $G$ and knows the entire topology of $G$. 
As we will later see, $k = \Omega(n)$ is an interesting regime in that several fundamental problems can be reduced to \kbroadcast with $k = \Omega(n)$. Many graph problems, such as the computation of distances and cuts, admit sparsifiers, so approximate solutions for these problems can be obtained by first computing a sparsifier and then using \kbroadcast to let all nodes learn the sparsifier, where $k$ equals the number of edges in the sparsifier. After that, all nodes have enough information to approximately solve the problem locally.
We first formally define the notion of universal optimality.

\paragraph{Universal optimality.}

We follow the approach of~\cite{universally_optimal_stoc2021} to define a universally optimal algorithm. 
Informally, for a given problem $\Pi$, we say that an algorithm  $\mathcal{A}$ is \emph{universally optimal} if its complexity on any input instance is within a polylogarithmic factor of the complexity of the fastest algorithm $\mathcal{A}^\ast$ specifically designed for this input instance $s$.
In essence, a universally optimal algorithm achieves the best possible complexity on every input instance, up to a polylogarithmic factor. 

The formal definition of universal optimality is naturally problem-specific.
Given a problem $P=(S, I)$, split its input into a fixed part $S$ and a parametric input $I$. For example, for \kbroadcast, we fix the graph $G$ and the starting locations of all $k$ messages, yet the contents of the messages are arbitrary. For a given algorithm $\mathcal{A}$ solving $P$ and any possible state $s$ for $S$ and $i$ for $I$, denote by $t(a, s, i)$ the round complexity of $\mathcal{A}$ when run on $P$ with $S=s$ and $I=i$. An algorithm $\mathcal{A}$ is universally optimal w.r.t.~$P$ if, for any choice of $s$, the worst-case round complexity of $\mathcal{A}$ is at most that of the best algorithm $\mathcal{A}_s$ for solving $P$ which knows $s$ in advance, up to a polylogarithmic factor. Formally, for all possible $s$ and any algorithm $\mathcal{A}_s$, set $t(\mathcal{A}, s) = \max_{i}t(\mathcal{A}, s, i)$ and $t(\mathcal{A}_s) = \max_{i}t(\mathcal{A}_s, s, i)$, it holds that $t(\mathcal{A}, s) = \tilde{O}(t(\mathcal{A}_s))$. That is, one must fix a \emph{single} algorithm $\mathcal{A}$ that works for all $s$, yet $\mathcal{A}_s$ can be different for each $s$.

Concretely, for the \kbroadcast problem considered in this work, we not only allow $\mathcal{A}^*$ to know the entire graph topology but also allow $\mathcal{A}^*$ to know the initial positions of all $k$ messages to be broadcast. Our broadcast algorithm, which costs $O((n+k)/\lambda) \cdot \log n)$ rounds, is universally optimal in the regime $k = \Omega(n)$, as there is a simple information-theoretic $\Omega(k/\lambda)$ lower bound for any algorithm $\mathcal{A}^*$ that knows the graph topology and the initial distribution of the $k$ messages. We emphasize that our broadcast algorithm does not need to know the graph topology and the initial distribution of the $k$ messages.

Next, we demonstrate a straightforward $\Omega(k/\lambda)$ {lower bound} for the \kbroadcast problem. 
The lower bound is \emph{universal} in the sense that the lower bound applies to every graph $G$, and moreover, the lower bound has to hold even for algorithms that are tailor-made for $G$. In the lower bound, we do not even need to control the nodes which are the sources of the messages, but rather just that the message contents should be random bits.

\begin{theorem}[Universal lower bound for \kbroadcast]\label{thm:universalLB}
For any graph $G=(V,E)$, any value $k$, and any initial distribution of the $k$ messages, any algorithm that solves the \kbroadcast problem with probability at least $1/2$  requires $\Omega(k/\lambda)$ rounds, even if the graph topology and the initial distribution of the $k$ messages are known to the algorithm. 
\end{theorem}
\begin{proof}
    Let each message be a uniformly random string of $s = O(\log n)$ bits.
    We select $S \subseteq V$  such that $|E(S, V\setminus S)|=\lambda$. Such a set exists by the definition of $\lambda$. It must be the case that either $S$ or $V \setminus S$ contains at least half of the messages initially. By symmetry, we assume that at least $k/2$ messages are initially in $V \setminus S$. In order to transmit these messages from $V \setminus S$ to $S$ across the  $\lambda$ cut edges $E(S, V\setminus S)$, this requires \[\Omega\left(\frac{sk/2}{|E(S, V\setminus S)| \log n}\right) = \Omega(k/\lambda)\] rounds of communication, as the communication bandwidth of an edge per round is $O(\log n)$.

    More formally, let $w = O(\log n)$ be the communication bandwidth of an edge per round, and let $t$ be the round complexity of the algorithm. 
    We claim that the success probability of the algorithm is smaller than $1/2$ if $2tw\lambda < (sk/2) -4$, so we must have $t = \Omega(k/\lambda)$.

    Suppose $2tw\lambda < (sk/2) -4$.
    Let $\mathcal{B}$ be the collection of all possible choices of the messages sent between $V \setminus S$ and $S$ throughout the algorithm. We have $|\mathcal{B}| =  2^{2tw\lambda} < 2^{(sk/2)-4}$.
     Let $\mathcal{M}$ be the collection of all possible choices of the messages that are initially in $V \setminus S$. We have $|\mathcal{M}| \geq 2^{sk/2}$. 
    We call each $m \in \mathcal{M}$   \emph{good} if there exists an element $b \in \mathcal{B}$ such that the nodes in $S$ correctly recover $m$ from $b$ with probability at least $1/4$, conditioning on $b$ being the outcome of the algorithm. The total number of good elements in $\mathcal{M}$ is at most $4|\mathcal{B}|$ because the events for recovering distinct $m \in \mathcal{M}$ are disjoint. Hence the success probability of the algorithm is at most
    
    \[\frac{1 \cdot 4|\mathcal{B}| + (1/4) \cdot |\mathcal{M}|}{|\mathcal{M}|}  < \frac{ (1/4) \cdot |\mathcal{M}|  + (1/4) \cdot |\mathcal{M}|}{|\mathcal{M}|}   \leq \frac{1}{2}. \qedhere\]    
\end{proof}

Since for $k = \Omega(n)$ the running time of our algorithm is $\tilde{O}(k/\lambda)$, it is universally optimal for $k = \Omega(n)$ by \cref{thm:universalLB}.


\section{Applications} \label{sec:APSP_all}

In this section, we discuss applications of our broadcast algorithm for approximating distances and cuts in graphs. 
In \cref{sec:APSP} we present an unweighted approximate APSP algorithm.
In \cref{sec:wAPSP} we present an unweighted approximate APSP algorithm via spanner computation.
In \cref{sec:cuts} we present a cut approximation algorithm via spectral sparsifiers.
In \cref{sec:APSP_LB}, we demonstrate the optimality of our algorithms by showing nearly matching lower bounds.

\subsection{Unweighted APSP}\label{sec:APSP}

For unweighted APSP, we prove the following theorem.

\begin{theorem} \label{thm:APSP}
    There is an algorithm that computes a $(3,2)$-approximation for unweighted APSP in $\tilde{O}(n/\lambda)$ rounds \whp 
\end{theorem}

At a high level, our algorithm works as follows. First, we break the graph into $\tilde{O}(n/\delta)$ clusters of diameter $O(1)$ by choosing a random set of $\tilde{O}(n/\delta)$ centers and having each node join the cluster of a neighboring sampled node. As the minimum degree is $\delta$, \whp~each node indeed has a sampled neighbor.
Next, we run an APSP algorithm on the cluster graph. Since this graph has only $\tilde{O}(n/\delta)$ clusters, we can simulate an APSP algorithm on this graph in $\tilde{O}(n/\delta)$ rounds. Finally, to estimate the distance between two nodes $u$ and $v$, we use the distance between the cluster $C_u$ of $u$ and the cluster $C_v$ of $v$ in the cluster graph. If the distance between $C_u$ and $C_v$ is $k$, we estimate the distance between $u$ and $v$ with $3k+2$. We prove that this indeed gives a $(3,2)$-approximation. To implement the algorithm efficiently, all nodes should know the cluster $C_u$ for each node $u$. We can solve this task efficiently using our broadcast algorithm. We next describe the algorithm in detail.

\paragraph{Building a cluster graph.}

Our algorithm is as follows. 
Let $p =  (c  \ln n)/\delta$, where $c > 0$ is some sufficiently large constant. Each node samples itself to be a \emph{center} with probability $p$ independently.
Let $k$ be the total number of centers.
By a Chernoff bound,  $k = \tilde{O}(n/\delta)$ \whp.
In the subsequent discussion, we write $\{c_1, \ldots, c_k\}$ to denote the set of all centers. 
    
    Next, we prove that if the minimum degree is at least $\delta$, then \whp for each node $v$, at least one of its neighbors lies in $\{c_1, \ldots, c_k\}$. For every node $u$, the probability that it is chosen as a center is $p = (c \ln n)/\delta$. Since each node $v$ has at least $\delta$ neighbors, the probability that no neighbor of $v$ is chosen as a center is at most $\left(1 -  (c \ln n)/\delta\right)^\delta \leq e^{-c \ln n} = n^{-c}$. By a union bound over all nodes, we conclude that with probability at least $1 - n^{1-c}$, for each node in the graph, at least one of its neighbors lies in $\{c_1, \ldots, c_k\}$.
    
    Every node $v \in V \setminus  \{c_1, \ldots, c_k\}$ selects an arbitrary neighbor $s(v)$ such that $s(v) \in \{c_1, \ldots, c_k\}$. Every center $c_i$ selects itself in that $s(c_i) = c_i$. For each center $c_i$, we write $C_i$ to denote the set of nodes $v$ with $s(v) = c_i$. Now, imagine a virtual graph $G_c$ over the nodes $\{c_1, \ldots, c_k\}$, where there is an edge between $c_i$ and $c_j$ if there exist two nodes $v_1$ and $v_2$ such that $\{ v_1,v_2 \} \in E$,  $c_i = s(v_1)$, and $c_j = s(v_2)$. 

    \paragraph{Computing APSP on the cluster graph.}

    We next compute APSP on the cluster graph.

    \begin{lemma} 
    \label{apsp-cluster}
    We can solve the APSP problem on $G_c$ in $\tilde{O}(n/\delta)$ rounds of communication over $G$. 
    \end{lemma}

    \begin{proof}
    Observe that every center $c_i$ can learn all of its neighbors in $G_c$ in $O(k) = \tilde{O}(n/\delta)$ rounds, as follows. First, let each node $v$ broadcast $s(v)$ to all its neighbors, and then for each cluster $C_i$, let $c_i$ gather all the messages sent to $C_i$. The round complexity is $O(k)$ since the number of distinct messages is at most $k$.
    
    After learning this information, we simply simulate the APSP algorithm of Peleg, Roditty, and Tal~\cite{unweightedapsp}. Their algorithm first performs a depth-first search from an arbitrary node, obtaining timestamps $\pi(u)$ denoting when the search first reached node $u$. Then, every node begins a breadth-first search with a delay equal to $2\cdot \pi(u)$. They prove that this ensures that no node $v$ is reached by the breadth-first searches of two different nodes at the same time. 
    Also, since we only perform depth-first or breadth-first searches, we have the property that any node $u$ always sends the same message to every neighbor $v$ in a particular round. Then, observe that we can clearly simulate $c_i$ sending the same message $M_i$ to all of its neighbors $c_{j_1}, \ldots, c_{j_l}$ using $3$ rounds in the original graph $G$ as follows: $c_i$ sends $M_i$ to all nodes in its cluster, then those nodes send it to all their neighbors from a different cluster, and finally those nodes send $M_i$ to their centers. By the properties mentioned earlier, we will have that any node $u$ will always have at most one distinct message to send to its center, and hence we are done. 
    \end{proof} 

    \paragraph{Estimating the distances.} We next show that the distances computed can help us estimate all the distances in the graph. We denote by $d_G(u,v)$ the distance between $u$ and $v$ in the graph $G$.

    \begin{lemma}\label{lem:approx_ratio}
         Define $d(u,v) = d_G(u,v)$ and $d'(u,v) = 3 \cdot d_{G_c}(s(u),s(v)) + 2$, then $d'$ is a $(3,2)$-approximation for APSP in $G$. 
    \end{lemma}

    \begin{proof}
        We should prove that $d(u,v) \leq d'(u,v) \leq 3d(u,v)+2$. We start by proving the left inequality.
        This follows from the fact that there is a path between $u$ and $v$ of length at most $3 \cdot d_{G_c}(s(u),s(v)) + 2$, by going from $u$ to $s(u)$, from $s(u)$ to $s(v)$ and then from $s(v)$ to $v$. The first and last edges add an additive 2 to the approximation, where $d_G(s(u),s(v)) \leq 3 \cdot d_{G_c} (s(u),s(v))$ because any virtual path in $G_c$ can be converted to a path in $G$ by replacing every virtual edge with at most $3$ edges in $G$. This shows that $d(u,v) \leq 3 \cdot d_{G_c}(s(u),s(v)) + 2$, as needed.

        To complete the proof, we should show that $3 \cdot d_{G_c}(s(u),s(v)) + 2 \leq 3d(u,v) +2$, this follows from the fact that $d_{G_c}(s(u),s(v)) \leq d(u,v)$, which can be proved as follows. Assume to the contrary that $d(u,v) < d_{G_c} (s(u),s(v))$, then we could have used the shortest path between $u$ and $v$, which changes clusters at most $d(u,v)$ times, to obtain a path from $s(u)$ to $s(v)$ in $G_c$ with length at most $d(u,v)$, which is a contradiction to the definition of $d_{G_c}(s(u),s(v))$. 
    \end{proof}

\paragraph{Putting everything together.} Based on the above ingredients, we can now prove \cref{thm:APSP}.

\begin{proof}[Proof of \cref{thm:APSP}]
We start by building the virtual graph $G_c$ in one round: All the sampled nodes announced to their neighbors that they are centers, and then each node $v$ can select $s(v)$ locally.
    We then solve the APSP problem on $G_c$ in $\tilde{O}(n/\delta)$ rounds using \cref{apsp-cluster}. Next, every center $c_i$ can simply broadcast $d_{G_c}(c_i,c_j)$ for all $j \in [k]$ to every node $v \in C_i$ in $O(k) = \tilde{O}(n/\delta)$ rounds. Also, every node $v$ can broadcast $s(v)$ to all other nodes $u$ using the broadcast algorithm from \cref{thrm:main} in $\tilde{O}(n/\lambda)$ rounds. Finally, using these two pieces of information, every node $v$ is then locally able to compute $d'(u,v)$, which is a $(3,2)$-approximation by \cref{lem:approx_ratio}. 
\end{proof}

\subsection{Weighted APSP}\label{sec:wAPSP}

For weighted APSP, we prove the following results.

\begin{restatable}{theorem}{weightedAPSP}\label{thm:weightedAPSP}
    For any integer $k \geq 1$, there is an algorithm that computes a $(2k-1)$-approximation for weighted APSP in $\tilde{O}\left(n^{1+\frac{1}{k}}/\lambda\right)$ rounds \whp     
\end{restatable}

The proof of \cref{thm:weightedAPSP} is based on constructing a $(2k-1)$-spanner with  $\tilde{m} = O\left(k \cdot n^{1+\frac{1}{k}}\right)$ edges and broadcasting it to the network using our broadcast algorithm. 

\begin{proof}
An $\alpha$-\emph{spanner} of a graph $G$ is a subgraph $H$ such that $d_H(u,v) \leq \alpha \cdot d_G(u,v)$ for all nodes $u$ and $v$ in $G$.
We show how to obtain a $(2k-1)$-approximation for weighted APSP using our broadcast algorithm. First, we use the Baswana--Sen algorithm~\cite{weightedspanner} to obtain a $(2k-1)$-spanner with  $\tilde{m} = O\left(k \cdot n^{1+\frac{1}{k}}\right)$ edges in $O(k^2)$ rounds. We then simply use our broadcast algorithm of \cref{thrm:main} to let all nodes in the graph learn the entire spanner, from which they can calculate a $(2k-1)$-approximation for APSP. Learning the spanner requires broadcasting $\tilde{m} = O\left(k \cdot n^{1+\frac{1}{k}}\right)$ messages. By \cref{thrm:main}, the overall round complexity of our algorithm is\[O(k^2) + \tilde{O}(\tilde{m}/\lambda) = O(k^2) + \tilde{O}\left(k \cdot n^{1+\frac{1}{k}}/\lambda\right).\]
If $k = O(\log n)$, then the above round complexity can be simplified to $\tilde{O}\left(n^{1+\frac{1}{k}}/\lambda\right)$. Otherwise, we may replace $k$ with $O(\log n)$ to attain the desired round complexity $\tilde{O}\left(n^{1+\frac{1}{k}}/\lambda\right)$, and this change of variable improves the approximation ratio.
 \end{proof}


 \begin{corollary}\label{cor:weightedAPSP}
     There is an algorithm that computes an $O\left(\frac{\log n}{\log \log n}\right)$-approximation for weighted APSP in $\tilde{O}(n/\lambda)$ rounds \whp 
 \end{corollary}
 \begin{proof}
 This follows from \cref{thm:weightedAPSP} with $k = \left\lceil \frac{\log n}{\log \log n} \right\rceil$, in which case $\tilde{O}\left(n^{1+\frac{1}{k}}/\lambda\right) = \tilde{O}(n/\lambda)$.
 \end{proof}

\subsection{Cuts}\label{sec:cuts}

We use the algorithm of Koutis and Xu \cite{koutis2016simple} to approximate all the cuts in the graph.
For a set of nodes $S$, let $\operatorname{cut}_G(S) = \sum_{u \in S, v \in V \setminus S} w(u,v)$.
Koutis and Xu~\cite{koutis2016simple} showed an algorithm for constructing a spectral sparsifier, that in particular has the following implication. 

\begin{theorem}[Koutis--Xu~\cite{koutis2016simple}]\label{thm:spectral}
    Consider a graph $G = (V, E, w)$. There exists a distributed algorithm in \congest such that for any $\epsilon > 0$ outputs a graph $H = (V, E', w')$ in $\tilde{O}(1/\epsilon^2)$ rounds satisfying the following conditions.
\begin{enumerate}
    \item $(1-\epsilon) \operatorname{cut}_H(S) \leq \operatorname{cut}_G(S) \leq (1+\epsilon)\operatorname{cut}_H(S)$ for any $S \subseteq V$.
    \item The number of edges in $H$ is $\tilde{O}(n/\epsilon^2)$.
\end{enumerate}
\end{theorem}

The above theorem statement is adapted from \cite{DBLP:conf/wdag/AnagnostidesG21}, which also used the algorithm of Koutis and Xu~\cite{koutis2016simple} to approximate cuts. The Koutis--Xu algorithm~\cite{koutis2016simple} is more general and builds a spectral sparsifier.
Combining  \cref{thm:spectral,thrm:main}, we obtain the following result.

\begin{theorem}\label{thm:cut_approx}
There is an algorithm that estimates all the values $\operatorname{cut}_G(S)$ up to a $(1+\epsilon)$ factor in $\tilde{O}\left(n/(\lambda \epsilon^2)\right)$ rounds \whp    
\end{theorem}
\begin{proof}
    Since the sparsifier $H$ of \cref{thm:spectral} has $\tilde{O}(n/\epsilon^2)$ edges, we can broadcast it to the whole graph in $\tilde{O}\left(n/(\lambda \epsilon^2)\right)$ rounds using the broadcast algorithm of \cref{thrm:main}. After that, all nodes can estimate all the values $\operatorname{cut}_G(S)$ up to a $(1+\epsilon)$ factor.
\end{proof}

\subsection{Lower Bounds}\label{sec:APSP_LB}

Our algorithms in \cref{thm:APSP,thm:cut_approx,cor:weightedAPSP} are  \emph{universally optimal} in the following sense: To write down the estimates of all distances or all cut sizes, it is \emph{necessary} to first learn the list of all IDs in the graph, and there is a simple information-theoretic  $\Omega(n / \lambda)$ universal lower bound for learning the list of all IDs.

\begin{theorem}[Universal lower bound for learning IDs]\label{thm:universalLB2}
Let $G=(V,E)$ be any graph with edge connectivity $\lambda$. If $\{\ID(v) \, | \, v \in V\}$ is chosen as a uniformly random subset of $[n^c]$ for some constant $c > 1$, then it requires $\Omega(n / \lambda)$ rounds for all nodes in $G$ to learn the list of all identifiers with probability at least $1/2$.
\end{theorem}
\begin{proof}
The proof is very similar to the proof of \cref{thm:universalLB}, so we omit the tedious details and only discuss the differences here. We select $S \subseteq V$  such that $|E(S, V\setminus S)|=\lambda$ and $|S| \leq |V|/2$. We reveal the IDs of the nodes in $S$ and do the remaining analysis conditioning on these IDs. Let $\mathcal{M}$ be the collection of all possible choices of $\{\ID(v) \, | \, v \in V \setminus S\}$. The way we pick the IDs ensures that \[|\mathcal{M}| 
 = \genfrac(){0pt}{}{n^c - |S|}{|V \setminus S|}
  \geq \genfrac(){0pt}{}{n^c /2}{|V \setminus S|}
 \geq \genfrac(){0pt}{}{n^c /2}{n/2} = 2^{\Omega(n \log n)}.\]
From the analysis of  \cref{thm:universalLB}, to let the nodes in $S$ learn $\{\ID(v) \, | \, v \in V \setminus S\}$ correctly with probability at least $1/2$, it is necessary to use \[\Omega\left(\frac{\log |\mathcal{M}|}{|E(S, V\setminus S)| \log n}\right) = \Omega\left(\frac{n \log n}{\lambda \log n}\right) = \Omega\left(\frac{n}{\lambda}\right)\]
rounds of communication.
\end{proof}

Same as \cref{thm:universalLB}, the $\Omega(n / \lambda)$ lower bound of \cref{thm:universalLB2} is a \emph{universal} lower bound in that the lower bound applies to \emph{every graph} $G$. In light of \cref{thm:universalLB2}, any algorithm that solves approximate APSP in $o(n / \lambda)$ rounds for certain graphs must involve strange tricks that recompute the IDs of some nodes in these graphs or must work under the assumption that the list of all IDs was known to the algorithm.

Next, we show that $\Omega(n / \lambda)$ is still a lower bound for approximate \emph{weighted} APSP even without the random ID assumption and allowing the nodes to recompute their IDs. Unlike \cref{thm:universalLB2} which is a {universal} lower bound, the following $\Omega(n / \lambda)$ lower bound only applies to a special family of graphs.

\begin{theorem}[Lower bound for weighted APSP]\label{thm:apspLB}
Let $\alpha = \Omega(1)$.
     For any two positive integers $\lambda$  and 
 $n$ such that $n \geq \lambda+1$, there exists a graph $G$ whose edge connectivity is $\lambda$ such that $\Omega\left(\frac{n}{\lambda \log \alpha}\right)$ is a lower bound for $\alpha$-approximate {weighted} APSP with probability at least $1/2$ on $G$ where all the weights are integers in $[n^c]$ for some constant $c > 0$.
\end{theorem}
\begin{proof}
Let $k_{\max} = \Theta\left(\frac{\log n}{\log \alpha}\right)$ be the largest integer such that $(2\alpha)^{k_{\max}} < n^c$. For each integer $i \in [3,n]$, let $k_i$ be an integer selected uniformly at random from $[k_{\max}]$ independently.
Consider the following construction of a graph $G=(V,E)$ whose edge connectivity is $\lambda$. 
\begin{itemize}
    \item Let $V = \{v_1, v_2, \ldots, v_n\}$ be the node set of $G$.
    \item Connect $v_1$ to $v_2$ using an edge with weight $1$.
    \item Connect $v_1$ to nodes in $\{v_3, v_4, \ldots, v_{\lambda+1}\}$ using edges with weight $n^c$.
    \item Make $\{v_3, v_4, \ldots, v_n\}$ a clique using edges with weight $n^c$.
    \item Connect $v_2$ to nodes in $\{v_3, v_4, \ldots, v_n\}$ such that the weight of edge $\{v_2, v_i\}$ is $(2\alpha)^{k_i}$ for each $i \in [3,n]$.
\end{itemize}
For node $v_1$ to $\alpha$-approximate its distance to all other nodes, $v_1$ must learn the precisely the values of $\{k_3, k_4, \ldots, k_n\}$. 
As $v_1$ is incident to $\lambda$ edges, there is an information-theoretic $\Omega\left(\frac{k_{\max} \cdot (n-2)}{\lambda \log n}\right) = \Omega\left(\frac{n}{\lambda \log \alpha}\right)$ lower bound for this task.
\end{proof}

As a consequence, the round complexity $\tilde{O}(n/\lambda)$ of \cref{cor:weightedAPSP} is nearly optimal.

\bibliographystyle{alpha}
\bibliography{References}

\appendix

\section{An Alternative Approach for Low-Diameter Tree Packings}\label{sec:existential}

In this section, we discuss an alternative approach for low-diameter tree packings based on the techniques in \cite{DBLP:conf/icalp/ChuzhoyPT20}.
Recall that \cref{thm:tree_packing} implies that a tree packing of $\lambda$ spanning trees of diameter $O((n \log n)/\delta)$ with congestion $O(\log n)$ exists for any graph with edge connectivity $\lambda$ and minimum degree $\delta$ and can be computed in $O((n \log n)/\delta)$ rounds by BFS (\cref{lemma:construct_bfs_tree}). Here we give an alternative proof of the existential result.

 \begin{theorem}\label{thm:tree_packing_existential}
Let $G$ be any simple graph with edge connectivity $\lambda$ and minimum degree $\delta$.
There is a polynomial-time algorithm that computes at least $\lambda$ spanning trees of diameter $O((n \log n)/\delta)$ such that each edge in $G$ appears in $O(\log n)$ trees \whp. 
\end{theorem}

The proof of \cref{thm:tree_packing_existential} presented below was suggested by an anonymous reviewer, who kindly allowed us to include the proof in the paper.
Following~\cite{DBLP:conf/icalp/ChuzhoyPT20}, we say that a graph is $(k,d)$-\emph{connected} if any two distinct nodes $u$ and $v$ can be connected by at least $k$ edge-disjoint paths of length at most $d$. The following lemma, which applies to even multigraphs, was shown in~\cite{DBLP:conf/icalp/ChuzhoyPT20}.

\begin{lemma}[{Chuzhoy--Parter--Tan~\cite[Theorem 1.4]{DBLP:conf/icalp/ChuzhoyPT20}}]\label{lem:CPTtool}
There is a polynomial-time algorithm that computes at least $k$ spanning trees of diameter $O(d \log n)$ of any $(k,d)$-connected multigraph $G$ such that each edge in $G$ appears in $O(\log n)$ trees \whp.  
\end{lemma}

By \cref{lem:CPTtool},  to prove \cref{thm:tree_packing_existential}, we just need to show the following lemma.

\begin{lemma}\label{lem:kdconn}
Any simple graph $G$ with edge connectivity $\lambda$ and minimum degree $\delta$ is $(\lambda/5, 16n/\delta)$-connected.
\end{lemma}

Let $G$ be any graph with edge connectivity $\lambda$ and minimum degree $\delta$.
Let $u$ and $v$ be any two distinct nodes in $G$.
To prove \cref{lem:kdconn}, we select $\mathcal{P}=\{P_1, P_2, \ldots, P_s\}$ as a set of $s = \max\{0, \lfloor \lambda/2\rfloor - 2\}$ edge-disjoint $u$-$v$ paths such that removing the set of all edges in these paths does not disconnect $u$ and $v$. Moreover, among all such collections of edge-disjoint paths, we select $\mathcal{P}$ to minimize the sum of path lengths over all paths in $\mathcal{P}$. After fixing $\mathcal{P}$, we select $P^\diamond$ as a minimum-length $u$-$v$ path that does not edge-intersect the paths in $\mathcal{P}$.

We say that a node $w$ in $P^\diamond$ is \emph{congested} if $w$ belongs to more than $s/2$ paths in  $\mathcal{P}$.  We write $C$ to denote the set of congested nodes in $P^\diamond$, and we write $U$ to denote the set of remaining nodes in $P^\diamond$. We bound the length $|U| + |C| - 1$ of $P^\diamond$ by analyzing $|U|$ and $|C|$ separately in the following two claims.

\begin{claim}\label{clm:Usize}
     $|U| < 8n/\delta$.
\end{claim}
\begin{proof}
Suppose $|U| \geq 8n/\delta$.
For each node $w \in U$, let $E_w$ denote the set of edges incident to $w$ that do not belong to the paths in $\mathcal{P} \cup \{P^\diamond\}$. We have $|E_w| \geq  \delta - 2 \cdot (1+s/2) \geq  \delta - \lambda/2 \geq \delta - \delta/2 = \delta/2$. The two endpoints of each edge in $E_w$ must be $w$ and some node that is not in $P^\diamond$, since otherwise $P^\diamond$ can be shortened using this edge, contradicting our choice of $P^\diamond$. Consequently, $|\bigcup_{w \in U} E_w| \geq |U| \cdot \delta/2 \geq 4n$. Therefore, there exists a node $w^\ast$ that is not in $P^\diamond$ and is adjacent to at least four nodes in $U$ via the edges in $\bigcup_{w \in U} E_w$, meaning that $P^\diamond$ can be shortened using two of these edges, contradicting our choice of $P^\diamond$. Therefore, $|U| < 8n/\delta$.
\end{proof}

\begin{claim}\label{clm:Csize}
     $|C| \leq |U|+1$.
\end{claim}
\begin{proof}
Suppose $|C| > |U|+1$. Then there exists an edge $e=\{w,x\}$ in $P^\diamond$ whose both endpoints are in $C$, meaning that they both belong to more than $s/2$ paths in $\mathcal{P}$, so there exists a path $P' \in \mathcal{P}$ containing both endpoints of $e$. Let $P''$ be the result of shortening $P'$ by replacing the $w$-$x$ subpath of $P'$ with $e=\{w,x\}$. Let $\mathcal{P}'$ be the result of replacing $P'$ with $P''$ in the set $\mathcal{P}$. Since $\mathcal{P}'$ is smaller than $\mathcal{P}$ in terms of the sum of path lengths, the existence of $\mathcal{P}'$ contradicts our choice of $\mathcal{P}$, so we must have $|C| \leq |U|+1$. To see that removing the set of all edges in the paths in $\mathcal{P}'$ does not disconnect $u$ and $v$, we may consider the $u$-$v$ path resulting from replacing the edge $e$ of $P^\diamond$ with the subpath connecting $w$ and $x$ of $P'$.
\end{proof}

We are ready to prove \cref{lem:kdconn}.

\begin{proof}[Proof of \cref{lem:kdconn}]
To prove the lemma, we just need to show that   $\mathcal{P}\cup \{P^\diamond\}$ is a set of at least $\lambda/5$ edge-disjoint $u$-$v$ paths of length at most $16n/\delta$.
By its definition, the set $\mathcal{P}\cup \{P^\diamond\}$ contains $s+1 = \max\{1, \lfloor \lambda/2\rfloor - 1\}$ edge-disjoint $u$-$v$ paths. Observe that the inequality $\max\{1, \lfloor \lambda/2\rfloor - 1\} \geq \lambda/5$ holds for all positive integers $\lambda$, so the requirement on the number of paths is satisfied.

 By \cref{clm:Usize,clm:Csize}, the number of nodes in $P^\diamond$ is  $|U| + |C| \leq 1 + 2|U| < 1 + 16n/\delta$, meaning that the length of $P^\diamond$ is less than $16n/\delta$.
 Our choice of $\mathcal{P}$ implies that each path in $\mathcal{P}$ has a smaller length than $P^\diamond$, so the requirement on the length of paths is also satisfied.
\end{proof}

We prove \cref{thm:tree_packing_existential} by combining \cref{lem:kdconn,lem:CPTtool}.

\begin{proof}[Proof of \cref{thm:tree_packing_existential}]
By \cref{lem:kdconn}, $G$ is $(k,d)$-connected with $k = \lambda/5$ and $d = 16n/\delta$. The polynomial-time algorithm of \cref{lem:CPTtool} computes a collection $\mathcal{T}$ of at least $k$ spanning trees of diameter $O(d \log n) = O((n\log n)/\delta)$ such that each edge of $G$ appears in $O(\log n)$ trees in $\mathcal{T}$ \whp.
By creating five copies of each spanning tree in $\mathcal{T}$, all requirements in \cref{thm:tree_packing_existential} are met.
\end{proof}

\paragraph{Remark.} Unlike the results in~\cite{DBLP:conf/icalp/ChuzhoyPT20} which apply to multigraphs, \cref{thm:tree_packing,thm:tree_packing_existential} do not extend to multigraphs. To see this, consider the multigraph resulting from replacing each edge of an $n$-node path graph with $k=n^{0.1}$ parallel edges. The minimum degree of the graph is $\delta = k$.
The diameter of any spanning tree of the graph is $n-1$, which is significantly larger than $n/\delta = n^{0.9}$, so \cref{thm:tree_packing,thm:tree_packing_existential} do not work for this graph.

\paragraph{Comparison.} \cref{thm:tree_packing} is stronger than \cref{thm:tree_packing_existential} in the sense that \cref{thm:tree_packing_existential} can be established as a corollary of \cref{thm:tree_packing}. For the other direction, the existence of the decomposition of \cref{thm:tree_packing} does not follow from \cref{thm:tree_packing_existential}. Unlike \cref{thm:tree_packing}, the algorithm underlying \cref{thm:tree_packing_existential} does not seem to admit an efficient distributed implementation.

\section{A Lower Bound for Tree Packing} \label{appendix_lower_bound}

In this section, we show how to extend the following lower bound of Ghaffari and Kuhn~\cite{DBLP:conf/wdag/GhaffariK13} to tree packings with low congestion. 

\begin{theorem}[{Ghaffari--Kuhn~\cite[Theorem D.1]{DBLP:conf/wdag/GhaffariK13}}] \label{lower_bound_broadcast}
For any $\lambda \geq 1$, there exist unweighted simple graphs $G=(V,E)$ with edge connectivity at least $\lambda$ and diameter $D=O(\log{n})$ such that for two distinguished nodes $s$ and $t$, sending $K$ bits of information from $s$ to $t$ requires $\Omega(\min{ \{ K/\log^2{n}, n/\lambda\}})$ rounds in \congest.
\end{theorem}

Consider $k$ arbitrary algorithms $\mathcal{A}_1, \mathcal{A}_2, \ldots, \mathcal{A}_k$. Let $\congestion$ be the maximum number of messages passing through an edge when the $k$ algorithms are run together.  Let $\dilation$ be the maximum round complexity over these $k$ algorithms. The following result was shown by Ghaffari~\cite{ghaffari2015near}.

\begin{theorem}[Scheduling distributed algorithms~\cite{ghaffari2015near}]\label{lem:scheduling}
There is an algorithm that executes any collection of distributed algorithms $\mathcal{A}_1, \mathcal{A}_2, \ldots, \mathcal{A}_k$ together in $O(\congestion) + O(\dilation \cdot \log^2 n)$ rounds \whp.
\end{theorem}

We prove the following diameter lower bound for tree packings with small congestion.

\begin{theorem}
Let $\lambda \geq \log^4{n}$ and $\CCC \leq \lambda/\log^4{n}$.
There are unweighted $\lambda$-edge-connected graphs $G$ with diameter $O(\log{n})$ such that if we decompose the edges of $G$ into $\lambda$ spanning subgraphs where each edge appears in at most $\CCC$ subgraphs, at least one of the subgraphs has diameter $\tilde{\Omega}(n/\lambda)$.
\end{theorem}

\begin{proof}
Suppose the graph is decomposed into $\lambda$ spanning subgraphs with diameter at most $\DDD$ and congestion at most $\CCC$. Consider the task of sending $K=\Theta(k \log{n})$ bits of information from a node $s$ to a node $t$. To solve this task, we let each of the $\lambda$ spanning subgraphs be responsible for sending $K/\lambda = \Theta( (k/\lambda) \log{n})$ bits of information. To do so, we run the basic broadcast algorithm of \cref{very_old_k_broadcast} in each subgraph in parallel using the scheduling algorithm of \cref{lem:scheduling}.
The overall round complexity for sending $K=\Theta(k \log{n})$ bits of information from a node $s$ to a node $t$ is $O((\DDD + \CCC k/\lambda) \log^2{n})$ rounds \whp. 
By \cref{lower_bound_broadcast}, this task requires $\Omega(\min{ \{ k/\log{n}, n/\lambda\}})$ rounds in a family of $\lambda$-edge-connected graphs with diameter $O(\log{n})$. If we choose $k=n \log{n}/\lambda$, then the lower bound becomes $\Omega(n/\lambda)$. On the other hand, since $\CCC \leq \lambda/\log^4{n}$, we infer that $O((\CCC k /\lambda) \log^2{n}) = o(n/\lambda)$, implying that $\DDD = \Omega(n/(\lambda \cdot \log^2{n})).$
\end{proof}

\end{document}